\documentclass[10pt,conference]{IEEEtran}
\usepackage{graphicx}
\usepackage{amsthm}
\usepackage{epsfig}
\usepackage{latexsym}
\usepackage{amsfonts}
\usepackage{here}
\usepackage{rawfonts}
\usepackage[latin1]{inputenc}
\usepackage[T1]{fontenc}
\usepackage{calc}
\usepackage{capitalgreekitalic}
\usepackage{url}
\usepackage{enumerate}
\usepackage{color}
\usepackage[tbtags]{amsmath}
\usepackage{amssymb}
\usepackage{upref}
\usepackage{epic,eepic}
\usepackage{times}
\usepackage{dsfont}
\usepackage{comment}
\usepackage{cite}

\DeclareMathOperator*{\argmax}{maximize}

\newtheorem{theorem}{\bf Theorem}
\newtheorem{proposition}{\bf Proposition}

\newtheorem{definition}{\bf Definition}
\newtheorem{remark}{\bf Remark}

\newcounter{step}
\newlength{\totlinewidth}
  {\end{list}%
  \rule{\linewidth}{1pt}}
\newcounter{substep}

  {\end{list}}

\newlength{\aligntop}
\newlength{\alignbot}
 \makeatletter
\renewenvironment{align}{%
  \vspace{\aligntop}
  \start@align\@ne\st@rredfalse\m@ne
}{%
  \math@cr \black@\totwidth@
  \egroup
  \ifingather@
    \restorealignstate@
    \egroup
    \nonumber
    \ifnum0=`{\fi\iffalse}\fi
  \else
    $$%
  \fi
  \ignorespacesafterend%
  \vspace{\alignbot}\par\noindent
} \makeatother

\IEEEoverridecommandlockouts

\begin{document}
%\pagenumbering{gobble}% Remove page numbers (and reset to 1)
\clearpage
%\maketitle
%Title.
% ------
\title{\huge Context-Aware Small Cell Networks: How Social Metrics Improve Wireless Resource Allocation}
%
% Single address.
% ---------------
\author{{Omid Semiari$^{1}$}, Walid Saad$^{1}$, Stefan Valentin$^{2}$, Mehdi Bennis$^{3}$, and H. Vincent Poor$^{4}$\vspace*{0em}\\
\authorblockA{\small $^{1}$Wireless@VT, Bradley Department of Electrical and Computer Engineering, Virginia Tech, Blacksburg, VA, USA,\\ Email: \protect\url{osemiari@vt.edu, walids@vt.edu} \\
$^{2}$Huawei Technologies, Paris, France, Email: \protect\url{stefanv@ieee.org}\\
$^{3}$Center for Wireless Communications-CWC, University of Oulu, Finland, Email: \protect\url{bennis@ee.oulu.fi}\\
$^{4}$Electrical Engineering Department, Princeton University, Princeton, NJ, USA, Email: \protect\url{poor@princeton.edu}\\
%%$^\ddag$Department of Electrical
%%    Engineering, Stanford University, Stanford, USA\\
%%Email: \protect\url{bmaham@stanford.edu,
%%$^\ddag$Dept. of ECE, Indian Institute of Science, Bangalore 560012, India\\
%Email: \protect\url{o.semiari@umiami.edu, walid@miami.edu, stefan.valentin@alcatel-lucent.com,bennis@ee.oulu.fi,poor@princeton.edu}
}\vspace*{0em}
    \thanks{This research was supported by the U.S. National Science Foundation under Grants CNS-1460316 and CNS-1513697.}%
  }

%
% For example:
% ------------
%\address{School\\
%   Department\\
%   Address}
%
% Two addresses (uncomment and modify for two-address case).
% ----------------------------------------------------------
%\twoauthors
%  {A. Author-one, B. Author-two\sthanks{Thanks to XYZ agency for funding.}}
%   {School A-B\\
%   Department A-B\\
%   Address A-B}
%  {C. Author-three, D. Author-four\sthanks{The fourth author performed the work
%   while at ...}}
%   {School C-D\\
%   Department C-D\\
%   Address C-D}
%
%\tableofcontents
%\pdfbookmarks
%\ninept
%
\maketitle
\thispagestyle{empty}
\vspace{0cm}
\begin{abstract}
In this paper, a novel approach for optimizing and managing resource allocation in wireless small cell networks (SCNs) with device-to-device (D2D) communication is proposed. The proposed approach allows to jointly exploit both the wireless and social \emph{context} of wireless users for optimizing the overall allocation of resources and improving traffic offload in SCNs. This context-aware resource allocation problem is formulated as a matching game in which user equipments  (UEs) and resource blocks (RBs) rank one another, based on utility functions that capture both wireless and social metrics. Due to social inter-relations, this game is shown to belong to a class of matching games with peer effects. To solve this game, a novel, self-organizing algorithm is proposed, using which UEs and RBs can interact to decide on their desired allocation. The proposed algorithm is then proven to converge to a two-sided stable matching between UEs and RBs. The properties of the resulting stable outcome are then studied and assessed. Simulation results using real social data show that clustering of socially connected users allows to offload a substantially larger amount of traffic than the conventional context-unaware approach. These results show that exploiting social context has high practical relevance in saving resources on the wireless links and on the backhaul.
\end{abstract}
\vspace{-0cm}

\vspace{-0cm}
\section{Introduction}
\label{sec:intro}
The introduction of smartphones and tablets has led to the proliferation of bandwidth-intensive wireless services, such as multimedia streaming and social networking, that have strained the capacity of present-day wireless communication networks \cite{17}. This increasing trend led to the emergence of wireless small cell networks (SCNs) as a promising solution to meet the quality-of-service (QoS) requirements of such emerging wireless services\cite{1,2,14,70}. In SCNs, the main idea is to massively deploy small cell base stations (SCBSs) with relatively low transmit power, overlaid on existing cellular infrastructure. Small cells allow to increase the capacity and coverage of a wireless network by bringing the user equipments (UEs) closer to their serving base stations. Nonetheless, the deployment of small cells introduces new challenges in terms of interference management, resource allocation, and network modeling. These challenges stem from many key features of SCNs such as the unplanned SCBS distribution, limited coverage, dense SCBS deployment, and limited backhaul capacities, among others \cite{1,2,14,70,74,50,82,84}.

%In \cite{72}, the concept of overlay cognitive radios is proposed for femtocell networks in order to manage the interference imposed on the small cells.
The authors in \cite{74} proposed a control-based scheduler for traffic management at small cells. In \cite{50}, an optimization problem is solved at each cell to perform resource allocation while taking into account cell range expansion and offloading metrics.
%Using random graphs, the authors in \cite{} study both shared and split spectrum reuse for resource allocation in two tier small cell networks.
Most of these existing approaches mainly adopt centralized methods for resource allocation \cite{74,50}. Although interesting, such centralized approaches have several drawbacks since they assume the presence of a centralized controller for the small cells, depend on SCBSs cooperation, and require macro base station (MBS) coordination. However, resource allocation in SCNs needs to be decentralized, self-organizing, and computationally efficient; specifically when the number of small cells increases. In this regard, game theory has emerged as a popular tool to realize distributed approaches for wireless networks \cite{27,84,89,94,28,82}. In \cite{82}, the authors proposed a distributed resource allocation in the uplink of a two tier network, by posing the problem as a matching game. They solved the game using the Hungarian algorithm. In \cite{84}, the resource allocation in SCNs is formulated as an evolutionary game. In \cite{27}, the theory of one-to-one and many-to-one matching markets is extended for the resource allocation in wireless networks. In \cite{89}, the authors used matching theory to perform distributed scheduling at the downlink of a MIMO-OFDMA system. Other works that apply matching in some limited wireless settings are found in \cite{94} and \cite{28}. In fact, prior works do not handle the challenges of SCNs that are underlaid with device-to-device (D2D) connections and in which there is a need not only to manage interference, but to also account for redundant transmissions by exploiting the ability of D2D to provide popular content caching. The body of work in \cite{89,27,94,28,84,82} focuses on resource allocation while only accounting for classical physical layer metrics such as the SINR and is restricted to networks without D2D. In addition, it is based on the classical deferred acceptance algorithm which cannot be applied for scenarios with peer effects such as in our case. Context-aware resource allocation, as done in our paper, is a new design paradigm that can help to boost the performance of small cell networks and to exploit D2D for popular content distribution.

%In fact, the body of work in \cite{89,27,94,28,84,82} focuses on scenarios in which resource allocation depends only on SINR information or is restricted to the traditional macro-cellular systems \cite{89,27,28} and, thus, do not handle the challenges of SCNs.

Along with the use of SCNs for improving network performance, D2D communications has recently emerged as an interesting approach to provide proximity services to users of an SCN, thus assisting in further offload of the cellular system's traffic \cite{65,66,67,87,91}. Indeed, due to the evolution of numerous data centric applications, it is very likely that devices in proximity of one another tend to interact directly over the wireless spectrum. Communication of such neighboring devices via the infrastructure of the SCN (i.e., via SCBSs) is neither spectrum nor power efficient \cite{65}. In addition, SCNs are envisioned to have a capacity-limited backhaul \cite{14} and, thus, the use of underlaid D2D can help offload traffic from the SCNs' backhaul. In this respect, D2D communication over the cellular spectrum is viewed as an attractive candidate to handle these scenarios\cite{93,87,91}. D2D over cellular networks is significantly different from D2D over unlicensed bands such as WiFi or traditional short-range D2D via Zigbee and Bluetooth. Indeed, D2D over cellular allows longer ranges and  higher QoS, while also requiring to properly manage interference with cellular transmissions \cite{66,67}, and \cite{93}. Some of the main challenges associated with deploying the D2D technology include introducing proximity services that leverage D2D, managing the wireless resources in D2D deployments, and protecting such low power and vulnerable communication links from interference \cite{65,66,67,87}.

In addition to the conventional physical layer metrics to optimize the SCN's performance, modern UEs can offer a versatile range of information from higher network layers that could help to reap the prospective gains of D2D deployments in SCNs.
%Inherently, by exploring more information from its environment, the network can manage the resources more efficiently, suppress the interference, and increase the overall wireless QoS.
Such additional information, referred to as \emph{context information}, may include the data extracted from online social networks \cite{75,59,80,81}, the history of a user's throughput \cite{56}, prediction of a user's location \cite{58}, or the delay-throughput tolerance of the applications \cite{55}.
%In addition, context-aware approaches have been used in ad hoc networks to increase the reliability of the service \cite{58}.
The work in \cite{75} develops an analytical model for the epidemic information spreading among mobile users of an ad hoc network. In \cite{59}, the resource allocation in wireless LAN is defined as an optimization problem, taking the notion of social distance into account. The authors in \cite{80} and \cite{81} extend the work in \cite{59} by introducing new utility functions which again account for the social distance of users, extracted from the social graph. Existing context-aware works \cite{56,55, 58,75,59,80,81} are mostly tailored to macrocell networks, are based on centralized approaches, and do not address the SCN or D2D over SCN challenges. In addition, although the use of social networks has been demonstrated to be useful to improve wireless systems, most existing works such as \cite{75,59,80,81} are based solely on the physical aspects of the social network, e.g., centrality measures. Such notion of social context is insufficient to capture common interests. For example, the large number of friends of a user in a social network does not necessarily mean that such a user is influential enough to require more bandwidth. In contrast, there is a need to adopt a more holistic view for the social context by basing it on other social dimensions such as the actual interactions between users. Here, we note that, although another body of works such as \cite{61,85,32,33} has further explored the use of social metrics in networking applications, such works are not adequate for deployment in wireless cellular systems such as SCNs with D2D as they do not deal with issues such as interference and network offload.

The main contribution of this paper is to propose a novel, self-organizing, context-aware framework for optimizing resource allocation in D2D-enabled SCNs. We formulate the problem as a two-sided one-to-one matching game in which each UE is assigned to one resource block (RB). In this game, the UEs and SCBS-controlled RBs rank one another based on utilities that capture the social context of the users as well as the wireless physical layer metrics. The social context includes the information inferred from the social network profiles of the wireless users. This information is mainly based on the similarities between users' interests, activities, and their interactions such as tagging or wall posting. The proposed scheme allows to exploit the fact that users who are strongly connected in a social network are likely to request similar type of data over the physical wireless network. We show that the proposed game is a matching game with \emph{peer effects} in which the strategy of each player is affected by the decisions of its peers. This is in contrast to most existing works on matching theory for wireless networks \cite{84,89,27,94,28,82} that deal with conventional matching games in which there is no peer effect. To solve this context-aware resource management game, we propose a novel, self-organizing algorithm that allows to find a stable matching between users and RBs. We show that our proposed algorithm allows the SCBSs and UEs to interact and converge to a stable matching with manageable complexity. Simulation results using real traces are used to analyze the performance of the proposed approach.

The rest of the paper is organized as follows. Section \ref{sec:SM} describes the system model. Section \ref{sec:3} introduces the modeling of social context in wireless D2D-enabled SCNs. Section \ref{sec:utility} defines the problem as the matching game and Section \ref{sec: Algorithm} presents the proposed algorithm. Simulation results are analyzed in Section VI and conclusions are drawn in Section VII.
\vspace{-0cm}
\section{System Model}
\label{sec:SM}
Consider the downlink of an OFDMA small cell network with a set $\mathcal{L}$ of $L$ SCBSs randomly distributed within the network. The total bandwidth $B$ is divided into $N$ RBs in the set $\mathcal{N}$ and there are a total of $M$ active users with $\mathcal{M}$ being the set of all users. We consider a co-channel network deployment in which the total bandwidth is shared between all small cells. In this network, we assume that users can communicate directly via D2D communication links within the cellular band. Such D2D communications enhance the indoor coverage and helps to offload the small cell traffic. In our model, some users are chosen as a serving user equipment (SUE) that are allowed to serve other UEs via D2D communication. Let ${\mathcal{M}_s}$ be the set of $M_s$ SUEs and ${\mathcal{M}_u}$ be the set of $M_u$ non-serving UEs. Thus, $\mathcal{M}={\mathcal{M}_s}\cup {\mathcal{M}_u}$ and ${\mathcal{M}_s}\cap {\mathcal{M}_u}=\emptyset$. The criterion for SUE selection is discussed further in Section \ref{SC}. Moreover, let $\mathcal{K}={\mathcal{M}_s}\cup {\mathcal{L}}$ be the joint set of all SCBSs and SUEs with $|\mathcal{K}|=M_s+L$. Hereinafter, we use the term ``serving node (SN)'' to refer to either an SCBS or an SUE. Moreover, we refer to cellular links and D2D links, respectively, as SCBSs to UEs and SUEs to UEs links.

For resource allocation in D2D-enabled SCNs, one simple approach is to allow the SCBSs to share all the RBs with the SUEs. However, in such a scheme, the D2D communication links will be dominated by the interference from the SCBSs. To overcome this problem, we propose to divide the spectrum in such a way that no mutual interference occurs between SCBSs and SUEs. Consequently, the sources of interference and the SINR relations will differ at each RB $n \in \mathcal{N}$, depending on whether RB $n$ is reused by an SCBS or an SUE.
\begin{figure}[!t]
  \begin{center}
   \vspace{0cm}
    \includegraphics[width=8cm]{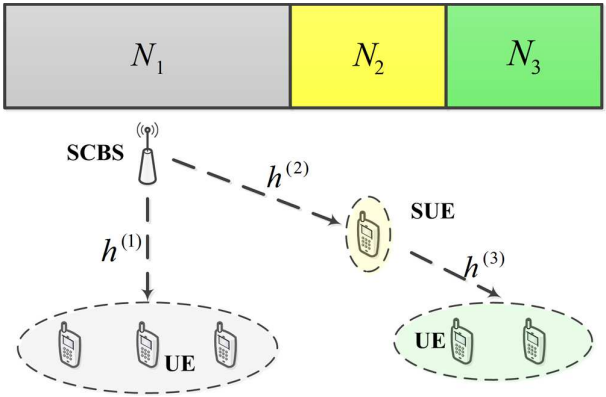}
    \vspace{-0cm}
    \caption{\label{fig:channel} Physical model for interference management in a D2D enabled SCN.}
  \end{center}\vspace{0cm}
\end{figure}
In this model, D2D links, SCBS to SUE links, and cellular links are separated in the frequency domain. Hence, the set of resource blocks $\mathcal{N}$, is divided into non-overlapping sets, namely, $\mathcal{N}_1$, $\mathcal{N}_2$, and $\mathcal{N}_3$ as shown in Fig. \ref{fig:channel}. $\mathcal{N}_1$ and $\mathcal{N}_2$ represent the set of $N_1$ and $N_2$ RBs dedicated to the direct links from SCBSs to UEs and SCBSs to SUEs, respectively. In addition, $\mathcal{N}_3$ is the set of $N_3$ dedicated RBs that are shared by all SUEs for D2D transmission. We let $h_{knm}$ be the channel state of subcarrier $n \in \mathcal{N}$ in the transmission from SN $k$ to user $m$. In this model, SCBSs may interfere with one another, since they share subbands $\mathcal{N}_1$ and $\mathcal{N}_2$. However, SCBSs will not interfere with D2D links as SUEs and SCBSs transmit on two different orthogonal bands. This encourages UEs to be served via D2D links which can improve the offloading capabilities of the network.

The achievable rate for the transmission between an SN $k \in \mathcal{K}$ and a user $m \in \mathcal{M}$ over RB $n \in \mathcal{N}_i$ is\vspace{-0cm}
\begin{align}\label{eq:1}
\Phi_{knm}(\gamma_{knm}^{(i)})=w_{n}\log(1+\gamma_{knm}^{(i)}),
\end{align}
where $w_{n}$ is the bandwidth of RB $n$ and $\gamma_{knm}^{(i)}$ is the instantaneous SINR for user $m$ from SN $k$ when using RB $n$. The superscript $i \in \left\{1,2,3\right\}$ indicates the set of RBs to which RB $n$ belongs. For $i \in \left\{1,2\right\}$ we have\vspace{-0cm}
\begin{align}\label{eq:2a}
\gamma_{lnm}^{(i)}=\frac{p_{ln}h_{lnm}}{\sum_{l' \in \mathcal{L}, l'\neq l} p_{l'n}h_{l'nm}+\sigma^{2}},
\end{align}
where $p_{ln}$ denotes the transmit power of SCBS $l$ over RB $n$. Moreover, $m$ corresponds to an arbitrary UE and SUE, respectively, for $i=1$ and $i=2$. For the transmissions over $\mathcal{N}_3$, the SINR is given by:\vspace{-0cm}
\begin{align}\label{eq:2b}
\gamma_{m_snm}^{(3)}=\frac{p_{m_sn}h_{m_snm}}{\sum_{m_s' \in \mathcal{M}_s, m_s'\neq m_s} p_{m_s'n}h_{m_s'nm}+\sigma^{2}},
\end{align}\vspace{-0cm}
where $p_{m_sn}$ denotes the transmit power of SUE $m_s$ over RB $n$ and $\sigma^2$ is the variance of the receiver's Gaussian noise.

Given this model, one important problem is how to allocate the bandwidth resources to the wireless users. As discussed in Section \ref{sec:intro}, beyond power allocation and interference management techniques, we can boost the capacity of wireless networks by making the network better informed of its environment. Recent studies \cite{32,33,61} have shown that friends in social networks, e.g. Facebook, have many common interests and activities that define their so-called \emph{social tie}. Such social ties' strength could properly show how frequently people interact with their friends, share popular videos or pictures, or invite one another to activities of common interest. Therefore, such interrelationships can explain how often socially connected people request common contents \cite{75,59,80,81,87}. Observing such behavior is interesting for SCN resource allocation, since it motivates the possibility of serving a user directly by other users with shared interests over D2D communication, instead of requesting the content from SCBSs. Therefore, such scheme allows the network to decrease redundant transmissions and offload this traffic from the backhaul network. In fact, we are investigating a \emph{content distribution} model that allows the network to use certain devices as SUEs, to serve as ``caching points'' whose storage can be used to cache popular content via overlay D2D. Therefore, our model is not a classical cooperative communication or relaying system.
\begin{figure}[!t]
  \begin{center}
   \vspace{0cm}
   \includegraphics[width=\columnwidth]{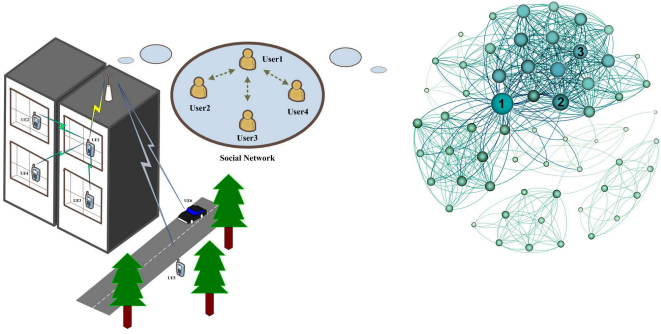}
    \vspace{-.3cm}
    \caption{\label{fig1} A schematic of a D2D enabled SCN which exploits the context information underneath the social network. The network in the right hand side shows the Facebook friendship graph of user 4.
    }
  \end{center}\vspace{0cm}
\end{figure}

For example, consider the scenario shown in Fig. \ref{fig1}, where a group of friends, e.g. students in a dorm or coworkers of a company spend a significant amount of time each day in neighboring rooms. Beyond this physical closeness, there is a social relationship between users at a higher layer that can underline their common interests in various topics such as sports or media. Since these users might have mutual interests, they are likely to be interested in common contents. Hence, although the formation of social ties is an application oriented metric, however, it strongly impacts how they access their wireless services, thus, directly impacting resource allocation. One illustrative example is the case in which one user, say user $1$, shares a certain video on the social network, which, in turn, will be viewed by some of its friends, due to the mutual interests. Hence, those users could be served using data that may be cached at UE $1$, directly through the D2D link. Clearly, by knowing the social ties between the users, the network will, on the one hand, be able to avoid multiple transmission of the same data and, on the other hand, will be able to allocate additional resources to the users outside the social group. In this work, we use the term \textit{traffic offload} to refer to the reduction of redundant transmissions from SCBSs to UEs that can be obtained by exploiting D2D links between SUEs and UEs. Such traffic offload will alleviate the traffic on the backhaul-constrained SCBSs while also allowing to service additional users over the SCBSs' RBs \cite{91}.

With this in mind, we propose to exploit, jointly with conventional channel information, the users' social interrelationships in order to optimize resource allocation in D2D-enabled SCNs.
%Indeed, we benefit from the contiguous nature of information spreading and common interests between friends that make users, request for highly correlated wireless data \cite{61,62,85,86,81,59,87}.
The resource allocation can be posed as an optimization problem in which RBs are assigned to UEs ($\xi^{\star}:\mathcal{N}\to\mathcal{M}$) such that the overall sum utility of the network is maximized.
% \cite{90}.
Taking the social context into account, we can formulate the problem as
\begin{align}
&\argmax_{\xi^{\star}} \,\sum_{k\in \mathcal{K}}\sum_{n\in\mathcal{N}}\sum_{m\in\mathcal{M}} \xi_{knm}\Omega_m (\Phi_{knm}(\gamma_{knm}),\boldsymbol{Z}),\label{opt11}\\\vspace{-.2cm}
&\text{subject to}\,\,\,\,\sum_{k\in \mathcal{K}}\sum_{m\in\mathcal{M}}\xi_{knm} \leq 1, \,\,\,\,\,   \forall n \in \mathcal{N},\label{opt12}\\\vspace{-.2cm}
&\,\,\,\,\,\,\,\,\,\,\,\,\,\,\,\,\,\,\,\,\,\,\,\,\,\,\,\sum_{k\in \mathcal{K}}\sum_{n\in\mathcal{N}}\xi_{knm} \leq 1, \,\,\,\,\,\,\,\,   \forall m \in \mathcal{M},\label{opt13}\\\vspace{-.2cm}
&\,\,\,\,\,\,\,\,\,\,\,\,\,\,\,\,\,\,\,\,\,\,\,\,\,\,\,\xi_{knm} \in \{0,1\},\vspace{-.1cm}\label{opt15}
\end{align}
where $\boldsymbol{Z}$ is a matrix that captures the social tie strength between every user pair and will be formally defined in Section \ref{sec:3}. Moreover, $\Omega_m(.)$ is the utility of user $m$ which is a function of achievable rates and social ties. If user $m$ is connected to an SCBS, $\Omega_m(.)$ simply represents the achievable rate of the link. If user $m$ is connected to an SUE, $\Omega_m(.)$ is the sum of the link's achievable rate, plus a term that determines how much user $m$ is socially connected to the cluster. The optimization problem in (\ref{opt11}) aims to maximize the sum utility of all users. The constraint in (\ref{opt12}) ensures that each RB is assigned to only one user, and the constraint in (\ref{opt13}) ensures that each user is assigned to one RB.

Due to the unplanned deployment of backhaul-constrained SCBSs and the limited possibilities for SCBS coordination \cite{14}, our goal is to develop a \emph{self-organizing, decentralized} resource allocation solution. This decentralized solution for the problem in (\ref{opt11})-(\ref{opt15}) will be addressed in depth in Section \ref{sec:utility}. Before doing so, we formally define the social tie strength in the next Section and explain how such context information could be extracted from the social networks.
\vspace{-0cm}
\section{Modeling Relationship Strength in Social Networks}\label{sec:3}\vspace{-0cm}
%Thus, instead of transmitting a certain file multiple times, the network can serve friends by themselves via direct D2D communication.
%To this end, our key goal is to exploit such social context for optimizing resource allocation in an SCN scenario. Naturally, this renders the problem more challenging due to the need for a resources allocation scheme that 1) does not compromise individual users' QoS, 2) is distributed and enables the network to be self-organizing, and 3) captures the social interrelationships among users along with the physical layer characteristics.
\subsection{Social Context in the Proposed SCN Model and SUE Choice}\label{SC}\vspace{-0cm}
Let $z_{ij}$ denote the social tie strength between two UEs $i$ and $j$. We define the social tie as a metric that determines how strong the relationship of two users is as inferred from the social network. This metric should then be incorporated into a proper utility function to be used in the context-aware resource allocation problem in (\ref{opt11}). In order to benefit from caching at the edge, popular contents must be cached at UEs that are chosen to serve as SUEs. Here, we assume that a user is chosen as SUE $m_s$ if its total social influence $I_{m_s}=\sum_{m \in \mathcal{M}, m\neq m_s} z_{m_sm}$ is larger than other users\footnote{Without loss of generality, other approaches for selecting SUE can also be accommodated.}. $I_{m_s}$ can be interpreted as a weighted degree of $m_s$ in a social network graph where the edge weights are determined by the $z$ term. We note that network operators need to provide some form of reward and incentive mechanisms to their users so that they act as SUEs. We can now define the notion of a \textit{social cluster}\vspace{-.2cm}
\begin{definition}
A \emph{social cluster (SC)} is defined as a set, $\mathcal{C}_{m_s}$, composed of an SUE $m_s$ and all the UEs which are connected to $m_s$ via D2D links.
\end{definition}\vspace{-0cm} We use the term \emph{social} here to emphasize that the social relationships of the users affect the formation of the cluster, as will be elaborated in Section \ref{sec:utility}. Due to the social effects, we can make the following observations: 1) a UE $m$ is encouraged by its friends to join the same SC in order to form socially stronger clusters, 2) SUEs with larger clusters (more assigned UEs) must get higher quality links from the $\mathcal{N}_2$ set, since the quality of the link from SCBS to SUE indirectly affects the quality of the direct D2D links, and 3) to improve offloading, SCBSs have an incentive to encourage UEs, with at least one friend as SUE, to use D2D links.

Such \emph{peer effects} motivate the need for an advanced model that can accurately define the strength of ties, $z_{ij}$. In \cite{63}, a graphical model is proposed to learn the strength of ties among a set of Facebook users which is suitable for our model. In particular, based on the homophily property, it is observed that the stronger the tie, the higher the similarity \cite{33}. Therefore, if two users have more attribute similarities in their profiles, e.g., the common groups that two UEs are members of, or the geographical locations, then their relationship is stronger. The relationship strength can be modeled as a \textit{hidden effect} of profile similarities and inferred via statistical learning concepts \cite{32,33}. In the following, we review a learning model based on \cite{63} that allows to understand how we can find the strength of ties, $z$ from a given social network dataset.
\vspace{-0cm}
\subsection{Learning Model}\label{learning model}
We note that the strength of the social relationship between two user impacts the nature and frequency of online interactions between a pair of users. Moreover, users naturally invest more of their resources (e.g., time) to build and maintain the relationships that they deem more important \cite{63}. Hence, as the relationship becomes stronger, it is more likely that a certain type of interaction will take place between the pair of users. In this way, we can model the relationship strength as the hidden cause of user interactions.

Formally, let $\boldsymbol{x}_i$ and $\boldsymbol{x}_j$ be, respectively, the attribute vectors of two UEs $i$ and $j$. An attribute vector is a vector that includes some of a user's social profile information, such as the user's age, political view, major, or the level of education. The relationship strength between $i$ and $j $ can be defined as a latent variable $z_{ij}$ which will be inferred for every pair of users. In addition, let $\boldsymbol{y}_{ij}$ be the vector of interactions whose elements $y_{ij,f}, f= 1,..., F$ are the $F$ different interactions considered between $i$ and $j$. Essentially, the interactions include the activities of UE $i$ that involves UE $j$, e.g., tagging one another or posting on each others' walls. This variable is assumed to be binary such that $y_{ij,f}=1$ if this interaction has occurred between UE $i$ and UE $j$ and $y_{ij,f}=0$, otherwise. Furthermore, the vector $\boldsymbol{e}_{ij,f}=[e_{ij,f}^1,...,e_{ij,f}^{\vartheta}]^T$ is defined for each interaction $f$ occurred between users $i$ and $j$.  This vector can show how much user $j$ is important for user $i$ to interact. For instance, if user $i$ has tagged user $j$ and assuming $\vartheta=1$, then $e_{ij,f}^{1}$ can be the overall number of users that user $i$ usually tags. Thus, smaller $e_{ij,f}^{1}$ implies stronger tie between users $i$ and $j$. This social model can be represented by a directed graphical model as shown in Fig. \ref{fig:model1}. In this model, $z_{ij}$ summarizes the profile similarities and interactions between users $i$ and $j$. However, it is not observable from users' profiles. Hence, we need to estimate $z_{ij}$ so as to maximize the overall observed data likelihood. To this end, the joint distribution of $z$ and $\boldsymbol{y}$ can be represented using general factorization:\vspace{-0cm}
\begin{align}\label{eq:a}
P(z_{ij},\boldsymbol{y}_{ij}|\boldsymbol{x}_{i},\boldsymbol{x}_{j})=P(z_{ij}|\boldsymbol{x}_{i},\boldsymbol{x}_{j})\prod_{f=1}^F P(y_{ij,f}|z_{ij}).
\end{align}
In order to infer the latent variables, we need to adopt the conditional probability of the relationship strength given the attribute similarities, i.e., $P(z_{ij}|\boldsymbol{x}_{i},\boldsymbol{x}_{j})$. In this regard, we consider the widely used Gaussian distribution \cite{63}\vspace{-0cm}
\begin{align}\label{eq:bb}
P(z_{ij}|\boldsymbol{x}_{i},\boldsymbol{x}_{j})=\mathcal{N}(\boldsymbol{w}^T \zeta(\boldsymbol{x}_{i},\boldsymbol{x}_{j}),\upsilon),
\end{align}
where the similarity vector $\zeta(\boldsymbol{x}_{i},\boldsymbol{x}_{j})$ is the set of similarity measures taken on the pair of users $(i,j)$ and $\boldsymbol{w}$ denotes the vector of parameters of the model in (\ref{eq:bb}).
\begin{figure}[!t]
  \begin{center}
   \vspace{0cm}
   \includegraphics[width=4cm]{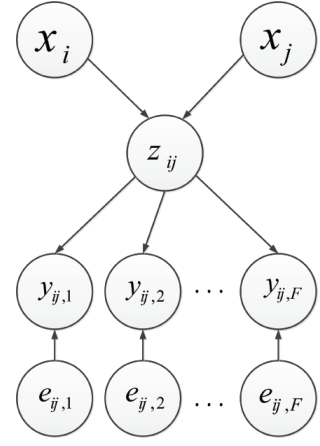}
    \vspace{-0cm}
    \caption{\label{fig:model1} The graphical representation of the social tie strength model \cite{63}.}
  \end{center}\vspace{0cm}
\end{figure}

Now, in order to completely describe the joint distribution in (\ref{eq:a}), the conditional probability of $y_{ij,f}$ given $z_{ij}$ and $e_{ij,f}$, can be modeled by using the logistic function:
\begin{align}\label{eq:c}
P(y_{ij,f}=1|\boldsymbol{u}_{ij,f})=\frac{1}{1+e^{-(\varrho_f^T \boldsymbol{u}_{ij,f})}},
\end{align}
where $\boldsymbol{u}_{ij,f}=[\boldsymbol{e}_{ij,f},z_{ij}]^T$. Moreover, $\boldsymbol{\varrho}_f=[\varrho_{f,1}, \varrho_{f,2}, ... ,\varrho_{f,\vartheta+1}]^T$ are the parameters of the model in (\ref{eq:c}) that must be estimated. From Fig.~\ref{fig:model1} and given latent variable $z_{ij}$, all elements of $\boldsymbol{y}_{ij}$ become independent of each other. Let $\mathcal{D}=\left\{(i_1,j_1), (i_2,j_2), ..., (i_D,j_D)\right\}$ be the set of user sample pairs observed from the network. The variables $\boldsymbol{x}_{i}$, $\boldsymbol{x}_{j}$, $\boldsymbol{y}_{ij}$ and $\boldsymbol{e}_{ij,f}$ could all be extracted from the social network. Hence, conditioned to the attribute similarities and model parameters, we can write (\ref{eq:a}) as:
\begin{align}\label{eq:d}
&P(z_{ij},\boldsymbol{y}_{ij}|\boldsymbol{x}_{i},\boldsymbol{x}_{j},\boldsymbol{w},\boldsymbol{\varrho})=\\\nonumber
&\prod_{(i,j)\in \mathcal{D}}\left(P(z_{ij}|\boldsymbol{x}_{i},\boldsymbol{x}_{j},\boldsymbol{w})\prod_{f=1}^F P(y_{ij,f}|z_{ij},\boldsymbol{\varrho}_f)\right).
\end{align}
Here, by substituting (\ref{eq:bb}) and (\ref{eq:c}) in (\ref{eq:d}), the joint distribution can be written as:
\begin{align}\label{eq:e}
&P(z_{ij},\boldsymbol{y}_{ij}|\boldsymbol{x}_{i},\boldsymbol{x}_{j},\boldsymbol{w},\boldsymbol{\varrho})\propto\\\nonumber
&\prod_{(i,j)\in \mathcal{D}}\left(e^{-\frac{1}{2\upsilon}\left(\boldsymbol{w}^T\zeta_{ij}-z_{ij}\right)^2}\prod_{f=1}^F\frac{e^{-(\boldsymbol{\varrho}_f^T \boldsymbol{u}_{ij,f})(1-y_{ij,f})}}{1+e^{-(\boldsymbol{\varrho}_f^T \boldsymbol{u}_{ij,f})}}\right).
\end{align}
\vspace{0cm}
\subsection{Inference}\label{sec:inference}
Given the defined learning model, we can now infer the social tie strength between each arbitrary pair of users $(i,j)$. One way to estimate $z_{ij}$ is to adopt the approach of \cite{63} in which $z_{ij}$ is treated as a parameter. Essentially, we can find the point estimates $\hat{\boldsymbol{w}},\hat{\boldsymbol{\varrho}}, \hat{z} $ that maximize the likelihood $P(y,\hat{z},\hat{\boldsymbol{w}},\hat{\boldsymbol{\varrho}}|x)$. To avoid overfitting the training dataset for the model in (\ref{eq:bb}) and (\ref{eq:c}), regularizers $\lambda_{w}$ and $\lambda_{\varrho}$ will be used respectively for the parameters $\boldsymbol{w}$ and $\boldsymbol{\varrho}$ with Gaussian priors. Overfitting can occur if the size of the $\boldsymbol{w}$ vector is too large for the observed data from the attribute vector $\boldsymbol{x}$. Using (\ref{eq:e}), we have:
\begin{align}\label{eq:g}
&P(z_{ij},\boldsymbol{y}_{ij},\boldsymbol{w},\boldsymbol{\varrho}|\boldsymbol{x}_{i},\boldsymbol{x}_{j})=\\\nonumber
&P(z_{ij},\boldsymbol{y}_{ij}|\boldsymbol{x}_{i},\boldsymbol{x}_{j},\boldsymbol{w},\boldsymbol{\varrho})P(\boldsymbol{w},\boldsymbol{\varrho}|\boldsymbol{x}_{i},\boldsymbol{x}_{j}),
\end{align}
and since the model parameters  $\boldsymbol{\omega}$ and  $\boldsymbol{\varrho}$ are independent of one another, as well as of the attributes of the users, we can write the joint conditional distribution in (\ref{eq:g}) as
\begin{align}\label{eq:gg}
P(z_{ij},\boldsymbol{y}_{ij},\boldsymbol{w},\boldsymbol{\varrho}|\boldsymbol{x}_{i},\boldsymbol{x}_{j})=
P(z_{ij},\boldsymbol{y}_{ij}|\boldsymbol{x}_{i},\boldsymbol{x}_{j},\boldsymbol{w},\boldsymbol{\varrho})P(\boldsymbol{w})P(\boldsymbol{\varrho}).
\end{align}

and hence,\vspace{-0cm}
\begin{align}\label{eq:h}
&\log P(z_{ij},\boldsymbol{y}_{ij},\boldsymbol{w},\boldsymbol{\varrho}|\boldsymbol{x}_{i},\boldsymbol{x}_{j})=\sum_{(i,j)\in \mathcal{D}}\!\!\!\!\left(\!-\frac{1}{2\upsilon}\left(\boldsymbol{w}^T\zeta_{ij}-z_{ij}\right)^2\right)\!\!+\notag\\
&\!\!\sum_{(i,j)\in \mathcal{D}}\!\!\!\left(\sum_{f=1}^F \left(-(1-y_{ij,f})(\boldsymbol{\varrho}_f^T \boldsymbol{u}_{ij,f})\!-\!\log\left(1+e^{-(\boldsymbol{\varrho}_f^T \boldsymbol{u}_{ij,f})}\right)\right)\!\!\right)\notag\\
&-\frac{\lambda_{w}}{2}\boldsymbol{w}^T\boldsymbol{w}-\sum_{f=1}^F \frac{\lambda_{\varrho}}{2}\boldsymbol{\varrho}_f^T\boldsymbol{\varrho}_f+ C,
\end{align}
where $C$ is a constant. (\ref{eq:h}) is a concave function and, thus, the latent variable and parameters can be derived by using gradients of this function. Using the iterative Newton-Raphson algorithm, we can maximize the function in (\ref{eq:h}) and find the optimum values for $z_{ij}$, $\boldsymbol{\varrho}_{f}$, and $\boldsymbol{w}$. We denote $\boldsymbol{Z}$ as a $M \times M$ matrix, where $z_{ij}$ is the element of $i$-th row and $j$-th column.

Given the proposed wireless model of Section II along with the inferred social tie metric $\boldsymbol{Z}$, we are now able to account for the social interconnections among users in conjunction with the physical layer constrains of the D2D-enabled SCN. We next develop a matching-based approach to allocate resources in a social context aware SCN.

%A UE $m$ tends to be connected to an SUE $m_s$ if sufficient social tie is established between $m$ and $m_s$. As we discuss in this section, these SUE are known to be more influential and are more likely to establish higher quality links with the corresponding SCBS.
\vspace{-0cm}
\section{Context-Aware Resource Allocation As a Matching Game}\vspace{0cm}
\label{sec:utility}
Having defined the social context, our next goal is to solve the resource allocation problem in (\ref{opt11}). The problem given by (\ref{opt11}), subject to (\ref{opt12})-(\ref{opt15}), is a 0-1 integer programming, which is a satisfiability problem as such one of Karp's 21 NP-complete problems \cite{92}. Hence, it is difficult to solve this problem via classical optimization approaches. Moreover, for a large-scale SCN network with D2D communication, it is desirable to solve the context-aware resource allocation problem in (\ref{opt11})-(\ref{opt15}) using a decentralized, self-organizing approach in which the SCBSs and devices can interact and make resource allocation decisions based on their local information without relying on a centralized entity for coordination. In addition, owing to the need for exploiting context information, it is of interest to define individual SN or UE utilities that capture the locally available context at each node.

To this end, matching theory is a promising approach to perform decentralized resource management in wireless networks\cite{30,27,28,82,89}. The main benefit of matching theory is its ability to define individual utilities per UE and SNs as well as the available algorithmic implementations that allow to provide a largely decentralized and self-organizing solution to the resource allocation problem in (\ref{opt11})-(\ref{opt15}) while accounting for all the nodes' information. In essence, a matching game is a situation in which two sets of players must be assigned to one another, depending on their preferences. In a matching game, each player must rank the players in the opposing set using a \emph{preference relation} which captures this player's evaluation of the players in the other set. In the studied context-aware model, the preference relation can be based on a variety of metrics related to the social and wireless realms. In this regard, we formulate the proposed resource allocation problem in SCNs as a two-sided one-to-one matching game in which each SN-controlled RB $n \in \mathcal{N}$ will be assigned to at most one user $m \in \mathcal{M}$ and vice versa. Therefore, we can formulate the problem as a one-to-one matching game given by the tuple $(\mathcal{M}, \mathcal{N}, \succ_{\mathcal{M}}, \succ_{\mathcal{N}})$. Here, $\succ_{\mathcal{M}}=\{\succ_m\}_{m \in \mathcal{M}}$ and $\succ_{\mathcal{N}}=\{\succ_n\}_{n \in \mathcal{N}}$ denote, respectively, the set of the preference relation of users and RBs. We formally define the notion of a matching:
\vspace{-0cm}
\begin{definition} A \emph{matching} $\mu$ is defined as a function from the set $\mathcal{M} \cup \mathcal{N}$ into the set of $\mathcal{M} \cup \mathcal{N}$ such that
$m=\mu(n)$ if and only if $\mu(m)=n$.
\end{definition}\vspace{-0cm}
Let $V_{m}(.)$ and $U_{n}(.)$ denote, respectively, the utility function of UE $m$ and RB $n$. Given these utilities, we can say that a user $m$ prefers RB $n_1$ to $n_2$, if $V_{m}(n_1)>V_{m}(n_2)$. This preference is denoted by $n_1\succ_m n_2$. Similarly, an RB $n$ prefers UE $m_1$ to $m_2$, if $U_{n}(m_1)>U_{n}(m_2)$ and this is denoted by $m_1\succ_n m_2$.

\vspace{-0cm}
\subsection{Users' Preferences}\label{sec:UE}\vspace{-0cm}
Depending on whether an RB $n$ is offered by an SUE or an SCBS, the UEs will have different utility functions. Moreover, the utility of one player may depend on the matching of other players, due to the peer effects described in Section \ref{SC}. In this regard, let $a_{mm_s;\mu}=\left\{a_{mm_s}|\mu\right\}$ be a variable used to determine the existence of a D2D link between UE $m$ and SUE $m_s$, conditioned on the current matching $\mu$. Similarly, we use $a_{mm_s;\mu\mu'}=\left\{a_{mm_s}|\mu,\mu'\right\}$, to indicate whether a certain UE $m$ is a member of $\mathcal{C}_{m_s}$ in both matchings $\mu$ and $\mu'$. The motivation for this definition will be explained in Section \ref{sec:conv}. Before computing the utilities, UEs first obtain the corresponding channel response of each RB, from all SNs and form a $K\times~\!\!N$ channel coefficient matrix $\boldsymbol{H}_{m}$. Using $\boldsymbol{H}_{m}$, each user $m \in \mathcal{M}$ shapes its achievable data rate matrix, i.e. $\boldsymbol{\Phi}_{\text{$m$}}$, whose elements are given by (\ref{eq:1})-(\ref{eq:2b}). From (\ref{eq:1}) and (\ref{eq:2a}), the utilities of a UE and an SUE, respectively, for RBs $n_1 \in \mathcal{N}_1$ and $n_2 \in \mathcal{N}_2$, are given by:
\begin{align}\label{eq:8}
V_{mn_i,l}(\gamma_{lnm}^{(i)})=w_{n_i}\log\left(1+\gamma_{lnm}^{(i)}\right),
\end{align}
where $m$ corresponds to a UE ($i=1$) or an SUE ($i=2$). From (\ref{eq:2b}), the utility of a user $m$  using an RB $n \in \mathcal{N}_3$ is given by:
\begin{align}\label{eq:9}
&V_{mn,m_s}(\boldsymbol{Z},\gamma_{m_snm}^{(3)},\mu)=\\\nonumber
&w_{n}\log\left(1+\gamma_{m_snm}^{(3)}\right)\!\!+
\alpha_m\!\!\left(z_{mm_s}+\sum_{j\in \mathcal{M}_u\backslash m} {a_{jm_s;\mu\mu'}z_{mj}}\right),
\end{align}
where $\alpha_m$ is a weighting parameter that allows to control the impact of the social ties on the overall decision of the user. The utilities given by (\ref{eq:8}) state that UEs and SUEs, respectively, rank RBs $n_1 \in \mathcal{N}_1$ and $n_2 \in \mathcal{N}_2$ based only on the achievable rates. However, (\ref{eq:9}) captures the peer effects on the UEs' preferences for RBs $n \in \mathcal{N}_3$. A UE can benefit more from the mutual interests among the SC members, if the SC members are more socially connected with one another and with the UE. The second term in (\ref{eq:9}) implies that the UE prefers to join an SC that has stronger ties with it.
%The main rationale behind such a utility formulation is explained as follows. First, a UE would benefit more from the mutual interests among the SC members if they are more socially connected. Second, leveraging the social tie information would help to increase the privacy and security for D2D UEs, since the tie strength might be correlated to the trustfulness between SUE and UE \cite{87}.  The second term in (\ref{eq:9}) captures the social tie of UE $m$ with SUE $m_s$.
\vspace{-.2cm}
\subsection{SN-based RBs' Preferences}\label{sec:RB}\vspace{0cm}
The proposed matching game can be fully represented once the preference of each RB is defined. The decision of an RB $n \in \mathcal{N}$ is mainly controlled by the SCBS or SUE that is using it. Hence, in the proposed game, SCBSs and SUEs make decisions on behalf of their RBs. In the considered model, there are three groups of RBs, each of which has a different preference over the UEs. Similar to Section \ref{sec:UE}, we define a novel scheme at the RB side of the game, which is based on the information extracted from the corresponding social network. For the transmission between SCBS $l \in \mathcal{L}$ and UE $m \in \mathcal{M}_u$, over RB $n\in \mathcal{N}_1$, the utility of RB $n\in \mathcal{N}_1$ when choosing UE $m$ is given by:
\begin{align}\label{eq:11}
U_{nm,l}(\boldsymbol{Z},\gamma_{lnm}^{(1)})= w_{n}\log\left(1+\gamma_{lnm}^{(1)}\right)-\beta_n(\sum_{j\in \mathcal{M}_s} {z_{mj}}).
\end{align}
Hence, $m_1\succ_{n} m_2$ if and only if $U_{nm_1,l}>U_{nm_2,l}$. $\beta_n$ is a weighting parameter that controls the impact of the social context. The second term in the right hand side of (\ref{eq:11}) implies that RBs $n \in \mathcal{N}_1$ give less utility to UEs who can be served by an SUE.

In addition, the utility achieved by RB $n\in \mathcal{N}_2$ when selecting SUE $m_s$, $U_{nm_s,l}$ is given by
\begin{align}\label{eq:12}\vspace{-0cm}
U_{nm_s,l}(\boldsymbol{Z},\gamma_{lnm_s}^{(2)},\mu)= w_{n}\log\left(1+\gamma_{lnm_s}^{(2)}\right)+\nu_n \cdot X_{m_s},
\end{align}
where $X_{m_s}=\sum_{m \in \mathcal{M}_u} a_{mm_s;\mu}z_{mm_s}$ is the total cumulative social tie strength of a particular SUE $m_s$ and $\nu_n$ is a weighting parameter that controls the importance of the SC that the SUE has formed in RB $n$'s utility. The utility function in (\ref{eq:12}) promotes SUEs with higher social ties since they are likely to form larger SCs. Finally, the utility of RB $n\in \mathcal{N}_3$ for user $m$ via D2D link from SUE $m_s$ is given by\vspace{-0cm}
\begin{align}\label{eq:13}
U_{nm,m_s}(\gamma_{m_snm}^{(3)})= w_{n}\log\left(1+\gamma_{m_snm}^{(3)}\right)+\kappa_n z_{mm_s}.\vspace{-.3cm}
\end{align}
The utility function in (\ref{eq:13}) implies that UEs must be accepted based on both quality of the D2D link and social context. $\kappa_n$ is a weighting parameter and implies that how important the role of social tie is for RB $n \in \mathcal{N}_3$ of an SUE $m_s \in \mathcal{M}_s$ in accepting a UE $m \in \mathcal{M}_u$.
\vspace{-.1cm}
\section{Proposed Context-Aware Resource Allocation Algorithm}\vspace{-0cm}
\label{sec: Algorithm}
Given the formulated context-aware matching game, our goal is to find a \emph{stable matching}, which is one of the key solution concepts in matching theory \cite{30}. Let $\mathcal{A}(\mathcal{M},\mathcal{N})$ denote the set of all possible matchings, and $\mu(m,n)$ denote a subset of $\mathcal{A}(\mathcal{M},\mathcal{N})$, where $m$ and $n$ are matched together. Then, we define a stable matching as follows:
\begin{definition}
A pair $(m,n) \notin \mu$, where $m \in \mathcal{M}, n \in \mathcal{N}$ is said to be a \emph{blocking pair} for the matching $\mu$, if there is another matching $\mu' \in \mu(m,n)$, where $\mu'\,\succ_m\,\mu \,\,\,\text{and}\,\,\, \mu'\,\succ_n\,\mu$.
\begin{comment}
\begin{align}
\mu'\,\succ_m\,\mu \,\,\,\text{and}\,\,\, \mu'\,\succ_n\,\mu.\vspace{-.6cm}
%V_{mr;k}(\textbf{Z},\gamma_{krm},\mu)>V_{mr;k}(\textbf{Z},\gamma_{krm},\mu)
\end{align}
\end{comment}
A matching $\mu^*$ is \emph{stable} if and only if there is no blocking pair.\vspace{-.2cm}
\end{definition}
%In this model, a stable matching ensures that players cannot deviate from the results of the matching.
%Such notion of two-sided stability ensures that no RB-UE or RB-SUE pair would establish a new link which is not included in the outcome of the resource allocation.
A stable matching solution for resource allocation problem in (\ref{opt11})-(\ref{opt15}) ensures that after allocating resources to the users, no RB-UE or RB-SUE pair in SCN would benefit from replacing their current association with a new link. That is, no user can benefit by changing its assigned frequency resource and vice versa. For the proposed context-aware resource allocation game, we can make the following observation:\vspace{0cm}
\begin{remark}
The proposed SCN matching game has \emph{peer} effects.
\end{remark}\vspace{0cm}
The RBs and users in a context-aware resource allocation game may change their preferences as the game evolves. That is, the preference of one player may depend on the preferences of the other players. For instance, the peer effects introduced in Section \ref{SC} make the preferences of RBs and users interdependent, due to the social interrelationships among users. This type of game is known as the \emph{matching game with peer effects}, in which players have preference ordering over the set of all possible matchings $\mathcal{A}(\mathcal{M},\mathcal{N})$\cite{Baron1}. This is in contrast with the traditional matching games in which players have fixed preference ordering\cite{89,27,94,28,30}.

For traditional matching games such as in \cite{89,27,94,28,30}, one can use the deferred acceptance algorithm, originally introduced in \cite{79}, to find a stable matching. However, such an algorithm may not be able to converge to a stable matching when the game has peer effects \cite{30}, such as in the proposed context-aware resource allocation model. Therefore, there is a need to develop new algorithms, that significantly differ from existing applications of matching theory in wireless such as \cite{89,27,94,28}, so as to find the solution of the studied matching game.
\vspace{-0cm}
\subsection{Proposed Socially-Aware Resource Allocation Algorithm}
To solve the formulated matching game, we propose a novel algorithm for resource allocation in D2D-enabled SCNs. Table~\ref{tab:algo} shows the various stages of this proposed socially-aware resource allocation (SARA) algorithm that allows to solve the SCNs' matching game.

\vspace{-0cm}
The proposed algorithm is composed of four main stages: Stage 1 includes the matching of SUEs with RBs $n \in \mathcal{N}_2$, Stage 2 focuses on the matching of UEs with RBs $n \in \mathcal{N}_1\, \cup\, \mathcal{N}_3$, Stage 3 focuses on updating the SC information, and Stage 4 during which the actual downlink transmission occurs. Initially, the SCBSs use the knowledge of social ties among users to choose the SUEs as discussed in \ref{SC}. Each SCBS sends a proposal to its neighboring users that are deemed influential enough to be an SUE. UEs accept or reject the proposals and the SCBSs broadcast the set of SNs, $\mathcal{K}$, and the sets of RBs, $\mathcal{N}_i, i=1,2,3$.
\begin{table}[!t]
%\scriptsize
  \centering
  \caption{%\mycaption{%\vspace*{-1em}
    \vspace*{-0em}Proposed Social Context-Aware Resource Allocation Algorithm}\vspace*{-0em}
    \begin{tabular}{p{3.2in}}
      \hline \vspace*{-0em}
      \textbf{Inputs:}\,\,$\mathcal{L},\mathcal{M},\mathcal{N},\boldsymbol{H}, \boldsymbol{Z}$\\
\hspace*{1em}\textit{Initialize:}   \vspace*{0em}
SCBSs send proposal to high influential UEs to act as SUE. SCBSs broadcast the set $\mathcal{K}$ and announce the set of available RBs in $\mathcal{N}_1$, $\mathcal{N}_2$, and $\mathcal{N}_3$.
Initialize the set of $\mathcal{C}_{m_s}$ for each SUE as an empty set.\vspace*{-0cm}

\hspace*{0em} \textit{Stage 1:}\begin{itemize}\vspace*{-0em}
\item[] \hspace*{0em}(a) SUEs determine their preference ordering for RBs $n \in \mathcal{N}_2$, using (\ref{eq:8}).
\item[] \hspace*{0em}(b) RBs $n \in \mathcal{N}_2$ calculate utility of each SUE applicant using (\ref{eq:12}) for the current state of the matching.
\item[] \hspace*{0em}(c) SUEs apply for RBs $n \in \mathcal{N}_2$ and get accepted or rejected via the deferred acceptance algorithm.
\end{itemize}\vspace*{-0cm}
\hspace*{0em}\textit{Stage 2:}\begin{itemize}\vspace*{-0em}
\item[]  \hspace*{0em}(a) UEs apply for RBs $n \in \mathcal{N}_1 \text{ and } n \in \mathcal{N}_3$, using (\ref{eq:8}) and (\ref{eq:9}), respectively.
\item[]  \hspace*{0em}(b) RBs $n \in \mathcal{N}_1$ and $n \in \mathcal{N}_3$ calculate utility of each UE applicant using (\ref{eq:11}) and (\ref{eq:13}), respectively.
\item[]  \hspace*{0em}(c) UEs get accepted or rejected by RBs $n \in \mathcal{N}_1 \cup\mathcal{N}_3$ through deferred acceptance algorithm.
\end{itemize}\vspace*{-0cm}
\hspace*{0em}\textit{Stage 3:}\begin{itemize}\vspace*{-0em}
\item[]   \hspace*{0em}(a) Update SC information, $\mathcal{C}_{m_s}$ for $\forall m_s \in \mathcal{M}_s$.
\item[]  \hspace*{0em}(b) SUEs broadcast SC information of the current matching, i.e., $a_{mm_s;\mu}$ coefficients to their nearby UEs.
\end{itemize}
\hspace*{0em}\textbf{while} $\mathcal{C}_{m_s}, \forall m_s \in \mathcal{M}_s$ remain unchanged for two consecutive matchings\\
\hspace*{0em}\textit{repeat Stage 1 to Stage 3}\vspace*{0em}\\
\hspace*{0em} \textit{Stage 4:}  \begin{itemize}\vspace*{-0em}
\item[]  \hspace*{0em}(a) For any cluster member, SUE determines if the current requested data exists in its directory.
\item[]  \hspace*{0em}(b) Actual downlink transmission of data occurs from each RB to its matched SUE or UE.
\end{itemize}
\hspace*{0em}\textbf{Output:}\,\,Stable matching $\mu^*$\vspace*{0em}\\
   \hline
    \end{tabular}\label{tab:algo}\vspace{-0.5cm}
\end{table}

After initialization, each SUE applies for $n \in \mathcal{N}_2$ based on (\ref{eq:8}) and each RB accepts the most preferred UE and rejects other proposals based on the utilities defined in (\ref{eq:12}). Stage 1 terminates once each SUE is accepted by an RB or rejected by all its preferred RBs. This matching remains unchanged until SUEs update the cluster $\mathcal{C}_{m_s}, \forall m_s \in \mathcal{M}_s$, sets. However, the new context information changes the preferences of the RBs for SUEs based on (\ref{eq:12}). That is due to the fact that $X_{m_s}=\sum_{m \in \mathcal{M}_u} a_{mm_s;\mu}z_{mm_s}$ determines how much the members of an SUE's cluster are socially connected. A larger $X_{m_s}$ implies that the members can benefit more from D2D due to common interests in their requested data. Therefore, RBs $n \in \mathcal{N}_2$ prefer to be matched to an SUE with larger $X_{m_s}$. Due to the change in their preference ordering, SUEs and RBs $n \in \mathcal{N}_2$ need to repeat this stage once the context information is updated.

Following the first stage, UEs apply for $n \in \mathcal{N}_1$ or $n \in \mathcal{N}_3$, based on the utilities defined in (\ref{eq:8}) and (\ref{eq:9}), respectively. The SCBSs and SUEs controlling RBs $n \in \mathcal{N}_1$ and $n \in \mathcal{N}_3$, accept the UE that gives the higher utility based, respectively, on (\ref{eq:11}) and (\ref{eq:13}), and reject the rest of the applicants. As long as $\mathcal{C}_{m_s}, \forall m_s \in \mathcal{M}_s$, do not change, UEs have strict preference over RBs $n \in \mathcal{N}_1 \cup \mathcal{N}_3$ and vice versa. Stage $2$ ends once each UE is accepted by one RB or rejected by all RBs of its preference list.

All clusters are subject to change due to the peer effects in the matching game. Thus, players need to update their preferences based on the new $\mathcal{C}_{m_s}, \forall m_s \in \mathcal{M}_s,$ information resulted from Stage $2$. In Stage 3, each SUE $m_s$ updates the SC information, i.e. $\mathcal{C}_{m_s}, m_s \in \mathcal{M}_s$, based on the results of the current matching and broadcasts $\mathcal{C}_{m_s}$ set to its nearby UEs and the corresponding SCBS. According to this information, players sort their preferences conditioned on the current matching. The algorithm terminates, once the $\mathcal{C}_{m_s}, m_s \in \mathcal{M}_s$ sets do not change for two consecutive matchings.
%We show in Section \ref{sec:conv} that the results of the matching remain consistent, if this condition is held.
In the final stage, once the matching is complete, the downlink transmission of the UEs occurs, using the allocated resource blocks. This stage is essentially the actual communication stage in the D2D-enabled SCN.

\vspace*{-.1cm}
\subsection{Convergence and Stability of the Proposed Algorithm}\label{sec:conv}
In this subsection, we prove the stability of the algorithm proposed in Table \ref{tab:algo}. Prior to doing so, we make the following definition:\vspace*{-0cm}
\begin{definition}
Given the social interrelationship between UEs, an SC $\mathcal{C}_{m_s}$ is said to be \emph{S-stable}, if both of the following conditions are satisfied:\\
1) No UE $m$ outside the cluster $\mathcal{C}_{m_s}$ can join it. That is, for any $m \notin \mathcal{C}_{m_s}$ and $n \in \mathcal{N}_3$ belonging to $m_s$, there is no pair $(m,n) \notin \mu$ where $m\, \succ_n \,\mu(n) \,\, \text{and}\,\, n \,\succ_m \,\mu(m)$.\\
2) No UE $m$ inside the $\mathcal{C}_{m_s}$ can leave the cluster. That is, for any $m \in \mathcal{C}_{m_s}$ and $n \in \mathcal{N}_1 \cup \mathcal{N}_3$ that does not belong to $m_s$, there is no pair $(m,n) \notin \mu$ where $m\, \succ_n \,\mu(n) \,\, \text{and}\,\, n \,\succ_m \,\mu(m)$.\\
A matching is S-stable, if and only if all the clusters are S-stable.
\end{definition}

This notion of stability guarantees that the peer effects cannot make a UE outside the SCs join a cluster. In addition, the UEs inside SCs will not leave or change their clusters. However, it is not sufficient to ensure the required two-sided stability of the matching. Next, we discuss a property of the proposed game and show why the S-stability of SCs is non-trivial.\vspace*{-0.1cm}
%An S-stable matching implies that the effect of social relationships has been included in the matching. Therefore, from (\ref{eq:8})-(\ref{eq:13}) it  can be seen that if all SCs are socially stable, the peer effects can no more affect the preference list of UEs and RBs.
%This definition is useful to prove the two-sided stability of our proposed algorithm.
\begin{proposition}\label{pro1}
Given the information on the social clusters, once a UE $m$ is accepted by an RB of a particular SC, $\mu(m)$ will not reject $m$ in favor of any new applicant. However, this does not imply UE $m$ has no incentive to leave the cluster.
\end{proposition}
\begin{proof}
First, consider the game at its initial state where no UE has been assigned to any SUE (no SC is formed). That is, $\sum_{j\in \mathcal{M}_u\backslash m} {a_{jm_s;\mu\mu'}z_{mj}}=0$ and $V_{mn,m_s}=w_{n}\log\left(1+\gamma_{m_snm}^{(3)}\right)+\alpha_m z_{mm_s}$ for $\forall m \in \mathcal{M}_u$. Thus, from (\ref{eq:9}), a UE $m$ that applies for RB $n \in \mathcal{N}_3$ during the first matching will do so only due to the higher achievable rates plus social tie with the corresponding SUE. Without loss of generality, the preference ordering of UE $m$ is identified with two sets $\mathcal{U}$ and $\mathcal{W}$, where\vspace{0cm}
\begin{align}\label{eq:19}
\mathcal{U}&=\{\forall n' \in \mathcal{N}_1 |V_{mn',l}>V_{mn,m_s}\}\notag\\
&\cup \{\forall n' \in \mathcal{N}_3|V_{mn',m_s}>V_{mn,m_s}\};\notag\\
\mathcal{W}&=\{\forall n' \in \mathcal{N}_1 |V_{mn',l}<V_{mn,m_s}\}\notag\\
&\cup \{\forall n' \in \mathcal{N}_3|V_{mn',m_s}<V_{mn,m_s}\}.\vspace*{-0.5cm}
\end{align}
In (\ref{eq:19}), $\mathcal{U}$ and $\mathcal{W}$, represent the set of all RBs that are, respectively, more preferred and less preferred to UE $m$ than RB $n$. Thus, the preference ordering of UE $m$ can be denoted by $\mathcal{U}\succ_m n \succ_m \mathcal{W}$. In Stage 2, the process of acceptance or rejection of applicants is done in a manner analogous to the conventional deferred acceptance algorithm \cite{27}. Thus, we can ensure that, for a given SC setting, each $n \in \mathcal{N}_3$ has accepted, out of the applicants, the UE that
\begin{align}\label{eq:20}
\mu(n)&=\argmax_{m} U_{nm,m_s}(\gamma_{m_snm}^{(3)})\notag\\
 &=\argmax_{m} w_{n}\log\left(1+\gamma_{m_snm}^{(3)}\right)+\kappa_n \cdot z_{mm_s}.
\end{align}
%From (\ref{eq:9}) and (\ref{eq:13}), the respective utility of RBs and UEs are symmetric before any SC is formed. That is because they both use the same SINR information, $\gamma_{m_snm}^{(3)}$, and the social context information, $z_{mm_s}$.
Therefore, if another UE $m'$ applies for the RB $n$ during the next matching $\mu'$, then necessarily $n'=\mu(m')\in \mathcal{W}$, if UE $m'$ is not unmatched. This is due to the fact that, if $n'=\mu(m')\in \mathcal{U}$, then it means that UE $m'$ already satisfies (\ref{eq:20}) for RB $n'$, and hence, it will be accepted by $n'$, before applying to $n$.

Now, since $\mu(m')\in \mathcal{W}$, we can conclude that $U_{nm',m_s}<U_{nm,m_s}$; since otherwise, $\mu(m')=n$ which contradicts (\ref{eq:20}). In other words, the UEs who are accepted in the first iteration will not be rejected by their match as the game proceeds. We can hold the same argument for the next iterations. Nevertheless, this does not imply that UE $m$ necessarily stays in its cluster forever. To show this, we give an example.\\
\textbf{Example:}  With this example, we show that UEs may alter the S-stability of the clusters. Given two UEs $m$ and $m'$, two RBs $n \in \mathcal{N}_1$ and $n' \in \mathcal{N}_3$, and two SUEs $m_{s1}$ and $m_{s2}$, assume that $V_{mn;m_{s1}}>V_{mn';m_{s2}}$ before any SC is formed and let the current matching be $\mathcal{C}_{m_{s1}}=\{m_{s1},m\}$ with $\mu(m)=n$ and $\mathcal{C}_{m_{s2}}=\{m_{s2},m'\}$ with $\mu(m')=n'$. Once the SC information is updated for UEs, the utility of UE $m$ for RBs $n$ and $n'$ might change to $V_{mn;m_{s1}}<V_{mn';m_{s2}}$, due to its social tie with $m'$. Therefore, $m$ will leave its current cluster $\mathcal{C}_{m_{s1}}$ and will join $\mathcal{C}_{m_{s2}}$. Due to this move, both clusters are not considered S-stable. \end{proof}\vspace*{-0cm}
Based on the Proposition \ref{pro1}, we can show the S-stability of the proposed context-aware resource allocation game as follows:\vspace*{-0cm}
\begin{theorem}\label{pro2}
Each SC becomes S-stable after a finite number of iterations and, thus, the proposed algorithm in Table \ref{tab:algo} is guaranteed to converge.
\end{theorem}
\begin{proof}
From Proposition \ref{pro1}, each $n \in \mathcal{N}_3$ will not reject its current match $\mu(n)$, in favor of other applicants. Thus, a new UE $m$ can join an SC, only if it applies for an unmatched $n \in \mathcal{N}_3$. After each matching is done, more UEs will join SCs, due to the non-negative social effect they impose on one another. Eventually, there is a stage where no more UEs prefer to join clusters.

There can be only two possibilities if a UE $m$ with $\mu(m)\in \mathcal{N}_3$ can leave its cluster, and thus, alter the S-stability condition: \emph{Case 1:} there is an unmatched $n \in \mathcal{N}_1$ such that $n\succ_m \mu(m)$, and \emph{Case 2:} there is an unmatched $n \in \mathcal{N}_3$ corresponding to another cluster where, again, we have $n\succ_m \mu(m)$. Next, we show that even when either Case 1 or Case 2 occurs, the current SC will converge to an S-stable SC.

For the first case, we can check that UE $m$ has surely been rejected in previous matchings by an RB $n \in \mathcal{N}_1$ such that $n\succ_m \mu(m)$. Here, since $\mu(n)$ is matched to another player and RB $n$ is unmatched now, $m$ will be accepted by RB $n$. If $\mu(n)$ never applies again for RB $n$, then we can ensure that UE $m$ will not be rejected by $n$, since it satisfies $\mu(n)=\argmax_{m} U_{nm,l}(\gamma_{lnm}^{(1)})$. Thus, UE $m$ will not get back to its SC again. If $\mu(n)$ applies again to RB $n$, then RB $n$ will reject $m$. However, $\mu(n)$ cannot cycle between RB $n$ and another RB $n'$ for unlimited number of iterations. This is due to the fact that, if $\mu(n)$ oscillates between $n$ and another RB $n' \in \mathcal{N}_3$, the reason is due to the peer effect by a current SC member $m''$ who also oscillates between $n'$ and another RB. However, $\mu(n)$ will set $a_{\mu(n)m'';\mu\mu'}=0$, if $m''$ does not stay in SC for two consecutive matchings. Therefore, $\mu(n)$ has to finally decide between $n$ and $n'$ after finite iterations. Moreover, if an RB $n' \in \mathcal{N}_1$ is the reason due to which $\mu(n)$ cycles, then similarly we can show that $\mu(n')$ will have to stop the oscillation after a finite number of iterations. Consequently, $\mu(n)$ will stop oscillation.

Case 2 implies that some of the friends of UE $m$ has encouraged $m$ to join their cluster. Here, if UE $m$ joins the new SC, it will never prefer the previous SC to new one, unless some of its friends, say $m'$, leave the cluster. Then, UE $m$ ignores the peer effect of $m'$ by setting $a_{mm';\mu\mu'}=0$ and reorganizes its preference ordering. After a finite number of iterations, the friends' list who stay in the cluster will not change and, thus, UE $m$ stays in the new cluster or comes back to the previous cluster and will never cycle between those two. In addition, if the incentive for UE $m$ to join RB $n'$ of the new cluster does not stem from the friendship relationship, then, this implies that $n'\succ_{m}\mu(m)$ and UE $m$ has been rejected by $n'$ in previous matchings. Thus, similar to the previous case, we can observe that after finite iterations, the $\mu(n')$ stops oscillating, and, hence, UE $m$ will decide whether to stay with $\mu(m)$ or leave it.\vspace{-0cm}
\end{proof}
Given the results in Proposition \ref{pro1} and Theorem \ref{pro2}, we can now state the main result with regard to the two-sided stability of the matching.\vspace*{-0cm}
\begin{theorem}\label{theorem1}
The proposed algorithm in Table \ref{tab:algo} is guaranteed to reach a two-sided stable matching between users and RBs.\vspace{-.2cm}
\end{theorem}
\begin{proof}
In order to prove the two-sided stability, we need to show that there is no blocking pair $(m,n)$ or $(m_s,n)$ that meets either of the following:
\begin{align}
&(m,n) \notin \mu \,\,\,\,|\,\,\mu(m) \in \mathcal{N}_3 \,\,\,,\,\,\,n \in \mathcal{N}_1 \cup \mathcal{N}_3;\label{eq:22}\\
&(m_s,n) \notin \mu\,\,|\,\, m_s \in \mathcal{M}_s\,,\, n \in \mathcal{N}_2;\label{eq:23}\\
&(m,n) \notin \mu \,\,\,\,|\,\,\mu(m) \in \mathcal{N}_1 \,\,\,,\,\,\,n \in \mathcal{N}_1 \cup \mathcal{N}_3.\label{eq:24}
\end{align}
All SCs are S-stable once the algorithm terminates since, otherwise, the SC sets will be different from those in previous matching which contradicts the termination condition.

From Theorem \ref{pro2}, we can see that no UE $m$ having $\mu(m) \in \mathcal{N}_3$ would make a blocking pair with any RB $n_1 \in \mathcal{N}_1$ or $n_3 \in \mathcal{N}_3$ from another cluster, due to the S-stability of all clusters. In addition, the matching of RBs $n \in \mathcal{N}_3$ belonging to SUE $m_s$ and UEs in $\mathcal{C}_{m_s}$ is done through the deferred acceptance algorithm. Therefore, these players will not form a blocking pair. Consequently, there is no blocking pair that satisfies (\ref{eq:22}).

In addition, the preferences of RBs $n \in \mathcal{N}_2$ become strictly fixed, once the S-stability is satisfied at all clusters. Then, followed by the deferred acceptance algorithm, we ensure that the given matching between SUEs and RBs $n \in \mathcal{N}_2$ is stable and there is no blocking pair that satisfies (\ref{eq:23}).

Finally, no RB $n \in \mathcal{N}_3$ makes a blocking pair with any $m$ where $\mu(m) \in \mathcal{N}_1$, due to the S-stability of all clusters. Moreover, the matching of RBs $n \in \mathcal{N}_1$ and UEs $m$ where $\mu(m) \in \mathcal{N}_1$ is followed by the deferred acceptance algorithm. Therefore, these players will not form a blocking pair and there is no blocking pair that satisfies (\ref{eq:24}). Hence, the proposed SARA algorithm is guaranteed to reach a two-sided stable matching for all D2D and cellular links.
%In addition, no $r \in \mathcal{R}_1$ would form a blocking pair with any UE since on the one hand, no UE would leave its SC. On the other hand, $r \in \mathcal{R}_1$ has already chosen its most preferred applicant $m \notin c_{m_s}$ for all SCs through the DAA. Finally, for any pair $(m_s,r) \in \mu|m_s \in \mathcal{M}_s, r \in \mathcal{R}_2$, the resulting match is stable, since players are matched via deferred acceptance algorithm conditioned to S-stable SCs. Therefore, there is no blocking pair and the output of the algorithm is stable.
\end{proof}\vspace*{-0cm}
%\subsection{Implementation Remarks}
\subsection{Complexity Analysis of the Proposed Algorithm}
In order to analyze the computational complexity of the proposed algorithm, we can start by investigating the simple case in which the matching game has no peer effect, i.e., users and RBs have strict preference ordering. Here, we consider two cases: 1) when the number of UEs is less than the total number of RBs, i.e., $M_u\leq N_T$, where $N_T=N_1\times L+N_3\times M_s$ and 2) when the number of UEs is greater than the total number of RBs, i.e., $M_u> N_T$. Our goal is to analyze the worst case scenario, i.e., the maximum number of iterations and the maximum number of matching proposals sent from UEs to either SCBSs or SUEs (which relate to the messaging overhead). In each iteration, UEs send a proposal to their most preferred RB, and RBs receive the proposals, accept the most preferred one and reject the other UEs. Therefore, it is clear that the number of unmatched UEs at each iteration is equal or less than the number of unmatched UEs at previous iterations.

For the first case, once the algorithm converges, all the UEs are matched, since RBs prefer any UE to being unallocated. We can easily observe that the worst case happens, if all UEs have the same preference ordering. Hence, at the end of each iteration $t$, there are $M_u-t$ unmatched UEs. Therefore, the maximum number of iterations, $t_{max}$, is obtained when all the users are matched, i.e., $M_u-t_{max}=0$. Hence, the complexity is of the order $\mathcal{O}(M_u)$. Furthermore, at each iteration $t$, $M_u-t+1$ proposals are sent. Hence, the messaging overhead, $S_{max}$, is equal to:
\begin{align}
S_{max}=\sum_{t=1}^{t_{max}}\left(M_u-t+1\right)=\frac{M_u(M_u+1)}{2}.
\end{align}

Similarly for the second case, we know that once the algorithm converges, there are exactly $M_u-N_T$ unmatched users. Again, the worst case happens, if all UEs have the same preference ordering. Hence, at each iteration, only one UE gets accepted. Therefore, the maximum number of iterations, $t_{max}$, is obtained when there are $M_u-N_T$ unmatched users, i.e., the complexity is of the order $\mathcal{O}(N_T)$. The messaging overhead is equal to:
\begin{align}
&S_{max}=\sum_{t=1}^{t_{max}}\left(M_u-t+1\right)=\\\nonumber
&\sum_{t=1}^{N_T}\left(M_u-t+1\right)=(M_u+1)N_T-\frac{N_T(N_T+1)}{2}.
\end{align}
As we see from the above equations, for $M_u<N_T$, the complexity of the matching algorithm increases linearly with the number of users. In addition, the messaging overhead exhibits quadratic increase with respect to the number of users. Moreover, for $M_u>N_T$, the upperbound for complexity is independent of the number of users.

The complexity of the proposed context-aware algorithm will further depend on the social matrix $\boldsymbol{Z}$. However, for a given SCs, the complexity of our algorithm follows the above analysis. From Theorem 1, we know that SCs change only for a finite number of iterations. Hence, we can anticipate that the overall complexity of the context-aware approach be linearly proportional to the complexity of the context-unaware approach.
\section{Simulation Results and Analysis}
\subsection{Social Context Dataset and Simulation Parameters}
For the evaluation of our results, we first use the learning model introduced in Section \ref{sec:3}. We have computed the social tie matrix $\boldsymbol{Z}$, for a set of 80 users from the Facebook network. We have used the real dataset released by Stanford University \cite{78}, which is generated by surveying a number of volunteer Facebook users, known as ego nodes. For the selected ego node, the dataset contains 224 anonymized attributes for each user, including education, gender, location, language, and work, among others. In addition, the dataset specifies 32 anonymized circles and determines which subset of users are in which circle. Each circle is a group, composed of a subset of users, which can be thought as a high school institution or a company. For our simulations, we assume that if two users $i$ and $j$ are within at least one common circle, then they interact with each other on Facebook. Finally, in order to determine the tendency of user $i$ to interact with another user $j$, we consider the degree of user $i$ in the ego network, i.e., the number of friends of user $i$. The motivation behind this assumption can be explained as follows: if user $i$ picks user $j$ to interact with from a larger number of friends, this implies that user $j$ is more important to user $i$ than other users.

In order to select the SUEs, we choose the four users with the highest weighted degree from the ego network. Thus, in our simulations, we have $M_s=4$ and the weights are determined by the social tie strength $z$ for each pair of users. We consider $L=7$ SCBSs distributed randomly within a square area of $2 \,\text{km} \times 2 \,\text{km}$. The number of active RBs for SCBS to UE link, SCBS to SUE link, and SUE to UE link, respectively, are $N_1=5, N_2=3$, and $N_3=5$, unless stated otherwise. Here, we would like to note that depending on the channel state information and social ties among users, spectrum partitioning can be done dynamically and the proposed model is not limited to any specific resource partitioning. In this work, however, we assume that throughout the resource allocation, neither the social tie matrix $\boldsymbol{Z}$, nor the channel state information are changing and therefore, the partitions are static. Each RB is composed of $12$ consecutive subcarriers, each of which having a $15$ KHz bandwidth, according to 3GPP Rel-12 Standard \cite{3gpp09:lte}. The transmit power of SCBSs and SUEs are set to $2$ W and $10$ mW, respectively. The wireless channel experiences Rayleigh fading, with the propagation loss set to $3$. The receivers' noise is assumed Gaussian with zero mean and with variance equals to $-90$ dBm. The weighting parameters, $\alpha_m$, $\beta_n$, $\nu_n$ and $\kappa_n$ are set to half of the RB bandwidth. Throughout the simulations, the unmatched users are assigned a zero utility. All statistical results are averaged over a large number of independent runs for different locations and channel gains.

For comparison purposes, we compare our proposed approach with two centralized approaches: 1) the context-aware centralized solution which aims to maximize the overall utilities of all UEs and RBs, 2) the context-unaware centralized approach, that maximizes the throughput of the users. In simulation results, the centralized solutions refer to the linear programming relaxation of the original 0-1 integer programming problem in (\ref{opt11})-(\ref{opt15}) by letting  $0\leq\xi\leq1$. To obtain centralized solutions in our simulations, we used the YALMIP toolbox of MATLAB. In addition, we compare our results with the context-unaware distributed algorithm that is based on the one-to-one matching game similar to our proposed algorithm, however, no social context is incorporated. That is, UEs and RBs rank one another only based on the maximum SINR values. This benchmark algorithm is in line with some existing works such as \cite{89} and~\cite{27}.

In addition, we show the effect of context-awareness by comparing the offloaded traffic of both approaches as one of the main performance metrics in our results. We define the offloaded traffic as the number of users who will be served directly by data that is cached in a directory or folder at the level of the SUE. The users obtain this offloaded traffic via D2D communications without having to use the SCN's infrastructure, due to the correlation in their requests for content as discussed in Section \ref{sec:3}. To compute this offloaded traffic, we must find the probability $P_m(y_d=1)$ of requesting content that already exists in the directory of SUE $m_s$ by each UE $m \in \mathcal{C}_{m_s}$:\vspace*{-0cm}
\begin{align}\label{eq:24}
P_m(y_d=1)=\sum_{d \in \mathcal{D}_{m_s}}P_m(y_d=1|d \in \mathcal{D}_{m_s})P(d \in \mathcal{D}_{m_s}),
\end{align}
where $d$ and $\mathcal{D}_{m_s}$ denote, respectively, the requested file and the set of files of the SUE $m_s$'s directory. %Requesting a file from the directory can be interpreted as an interaction between the UE $m$ and SUE $m_s$, which depends on their social tie strength. Hence, to compute the probabilities $P_m(y_d=1|d \in \mathcal{D}_{m_s})$ in (\ref{eq:24}) we suggest to use the conditional distribution in (\ref{eq:c}).
The Prior information, $P(d \in \mathcal{D}_{m_s})$ in (\ref{eq:24}), depends on both the history of requested files in previous time slots and on the mutual social tie between all members of the SC. Finding a closed-form solution for (\ref{eq:24}) to model the correlation between users in the time domain is difficult. Thus, for simplicity, we assume that requesting a file from the directory of an SUE can be modeled as an interaction $y_d$ between the user and its SC set. Motivated by  (\ref{eq:c}), we model $P_m(y_d=1)$ as a function of cluster's average social tie $\bar{z}_{m_s}=\frac{1}{|\mathcal{C}_{m_s}|-1}\sum_{m \in \mathcal{C}_{m_s}}z_{mm_s}$, as follows\vspace*{-0.2cm}
\begin{align}\label{eq:25}
P_m(y_d=1)=\frac{1}{1+e^{(-\rho\bar{z}_{m_s})}}; \,\,\,\,\,\forall m \in \mathcal{C}_{m_s},
\end{align}
where $\rho$ is a constant normalizing parameter. Equation (\ref{eq:25}) allows to relax the dependencies between users by considering average social tie of the cluster.\vspace{-0cm}
%For comparison purposes, we consider two scenarios in which either the data cashing is incorporated in the context-unaware approach or not.
\subsection{Simulation Results and Discussions}
\begin{figure}[!t]
  \begin{center}
   \vspace{0cm}
    \includegraphics[width=\columnwidth]{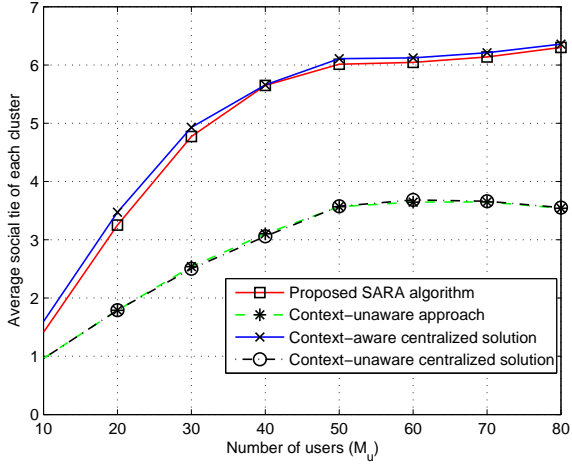}
    \vspace{-0.3cm}
    \caption{\label{fig6} Comparison of the average social tie of SCs for proposed SARA algorithm with other three approaches.}
  \end{center}\vspace{-.5cm}
\end{figure}

In Fig. \ref{fig6}, we show the average social tie in the final clusters resulting from the proposed approach and the other three algorithms. Clearly, from this figure, we can observe that the members of the clusters in average are much more socially connected in SARA compared to the context-unaware approach. Fig. \ref{fig6} shows that the average cluster social tie is up to $77 \%$ higher for the proposed SARA algorithm relative to the context-unaware scenario for $M_u = 70$ UEs. %By combining the results of Fig. \ref{fig5} along with Fig. \ref{fig6},
As the number of users increases and the SCBS resources become more scarce, UEs have no choice but to join the D2D clusters. However, Fig. \ref{fig6} shows that by using the proposed SARA algorithm, the UEs will be more socially connected within their clusters and can benefit more from the social interconnections of one another.
  \begin{figure}[!t]
  \begin{center}
   \vspace{0cm}
    \includegraphics[width=\columnwidth]{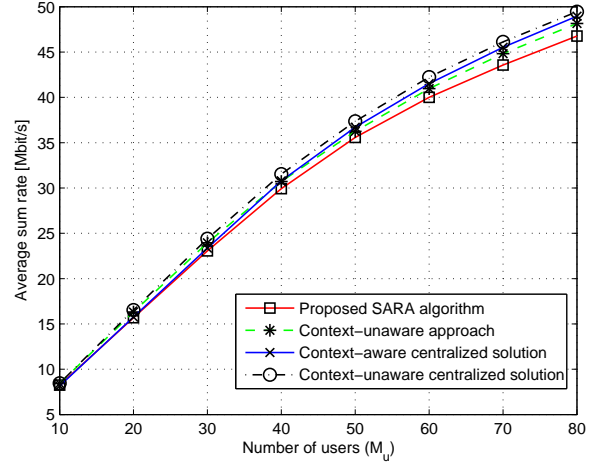}
    \vspace{-0.3cm}
    \caption{\label{fig8} Comparison of the average sum rate of the proposed SARA algorithm and other three approaches.}
  \end{center}\vspace{-.5cm}
\end{figure}

In Fig. \ref{fig8}, the average sum rate of all users is compared for the SARA algorithm with the distributed context-unaware approach and centralized approaches. It is not surprising to see that the average sum rate of our proposed approach is slightly below that of an algorithm that is focused on optimizing only the data rate. Fig. \ref{fig8} shows that the gap between the two algorithms does not exceed $3\%$ for all network sizes. However, as will be shown in Fig. \ref{fig9}, this small loss in average rate will be compensated by having more offloaded traffic in the downlink of the SCN. Interestingly, the rate performance of the proposed approach is very close to the centralized solution and the gape between the two algorithms does not exceed $4 \%$ for all network sizes.

Fig. \ref{fig9} compares the average offloaded traffic at each time slot for the proposed algorithm and other three approaches as the number of users varies. We note that for the context-unaware approaches, the contents stored in the SUE's directory do not depend on the average social tie of the social cluster. Considering $\rho=0$ in (\ref{eq:25}) allows us to model this independence. The average offloaded traffic determines how many of UEs in average could be served directly by data that is cached in the SUEs' directories. The offloaded traffic increases with the number of users, since more UEs will move to the D2D tier. However, this metric will saturate for the large network sizes, since the number of RBs $n \in \mathcal{N}_3$ is limited. We can see that the proposed approach achieves a very close performance to the
context-aware centralized solution, in terms of traffic offloads, for the network size with more than $M_u=40$ UEs. Moreover, we observe that the performance of the context-unaware approach coincides with the central solution. This is primarily due to the fact that context-unaware approach follows the deferred acceptance algorithm which is known to achieve optimal solution for the proposing players (i.e. users) \cite{79,27}. In Fig. \ref{fig9}, we can see that the proposed SARA algorithm outperforms the context-unaware approach by increasing the offloaded traffic for different network sizes. Fig. \ref{fig9} shows that the proposed SARA algorithm can offload up to $84 \%$ more traffic from the SCN's infrastructure when compared to the context-unaware approach for a network with $M_u=70$ UEs.
\begin{figure}[!t]
  \begin{center}
   \vspace{0cm}
    \includegraphics[width=\columnwidth]{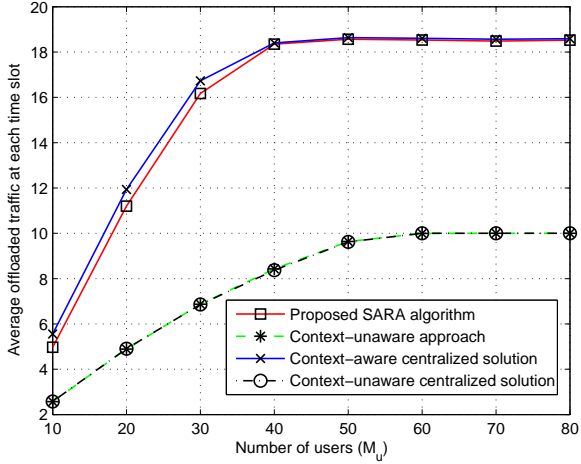}
    \vspace{-0.3cm}
    \caption{\label{fig9} Comparison of the average offloaded traffic between the proposed SARA algorithm and other three approaches. The $\rho=0.5$ and $\rho=0$ is assumed, respectively, for context-aware algorithms and context-unaware approaches.}
  \end{center}\vspace{-.5cm}
\end{figure}
 \begin{figure}[!t]
  \begin{center}
   \vspace{0cm}
    \includegraphics[width=\columnwidth]{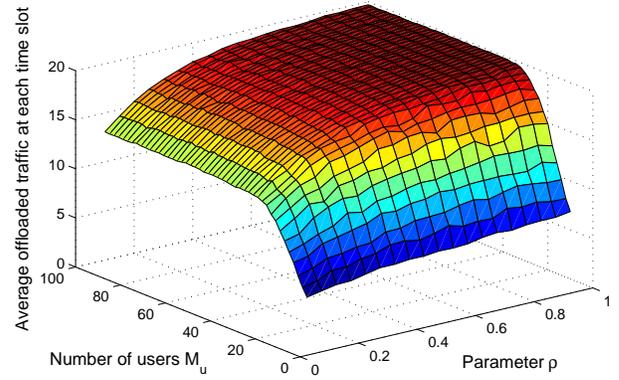}
    \vspace{-0.3cm}
    \caption{\label{fig10} The average offloaded traffic resulted by proposed SARA algorithm, versus different $M_u$ and different values of $\rho$.}
  \end{center}\vspace{-.5cm}
\end{figure}

In Fig. \ref{fig10}, we show the average offloaded traffic for different values of the $\rho$ parameter in (\ref{eq:25}) and different number of users. One interesting observation from Fig. \ref{fig10} is that the proposed SARA algorithm offloads more traffic as the network size grows and eventually saturates, since the number of RBs is fixed. In fact, Fig. \ref{fig10} implies that, at the cost of a slight reduction in the average rate of the matched users (see Fig. \ref{fig8}), the proposed SARA algorithm allows to admit more users into the network. Hence, the proposed approach can be used to optimize the tradeoff between serving more users in a congested network and the average rate of current users. By properly setting the control parameters in the utility functions, this tradeoff can be balanced according to the traffic of the network.
\begin{figure}[!t]
  \begin{center}
   \vspace{0cm}
    \includegraphics[width=\columnwidth]{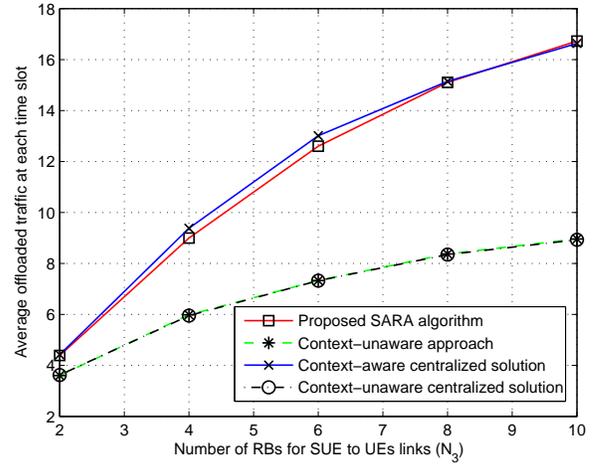}
    \vspace{-0.3cm}
    \caption{\label{fig11} Comparison of the average offloaded traffic for the proposed SARA algorithm with other three approaches, versus different number of RBs $n \in \mathcal{N}_3$. The values $M_u=30$ and $N_1=5$ is assumed and the parameter $\rho$ is set to $0.1$. }
  \end{center}\vspace{-.5cm}
\end{figure}

\begin{figure}[!t]
  \begin{center}
   \vspace{0cm}
    \includegraphics[width=\columnwidth]{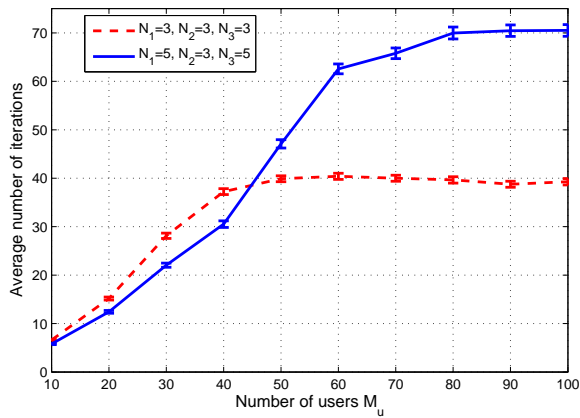}
    \vspace{-0.3cm}
    \caption{\label{fig7} Average number iterations of the proposed SARA algorithm vs network size, for different number of RBs. The error bars indicate the $95 \%$ confidence.}
  \end{center}\vspace{-.5cm}
\end{figure}

In Fig. \ref{fig11}, we show the average offloaded traffic resulting from the proposed SARA algorithm and other approaches as the number of RBs $n \in \mathcal{N}_3$ varies for a network with $M_u=30$ UEs. This figure shows that, as the number of SUE resources increases, the context-aware approach achieves better offload performance. That is due to the fact that more users with shared interests can join the same cluster which results to form stronger social clusters. In Fig. \ref{fig11}, we can see that the gap between the offloaded traffic of the proposed SARA algorithm and the context-unaware approach consistently increases as the network sizes grows. Fig. \ref{fig11} shows that the proposed SARA algorithm can offload up to $78 \%$ more traffic from the SCN's infrastructure when compared to the context-unaware approach for a network with $N_3=8$ RBs.

Fig. \ref{fig7} shows the average number of iterations resulting from the proposed SARA algorithm versus the network size $M_u$, for two different RB numbers. In this figure, we can see that, as the number of UEs increases, the average number of iterations increases due to the increase in the number of players. Fig. \ref{fig7} demonstrates that the proposed matching approach has a reasonable convergence time that does not exceed an average of $72$ iterations for all network sizes and with $N_1=5$, $N_2=3$, and $N_3=5$ number of RBs. The maximum of the average number of iterations reduces to $40$, when there are $N_1=3$, $N_2=3$, and $N_3=3$ RBs.
\vspace{-0cm}
\section{Conclusion}\vspace{-0cm}
In this paper, we have presented a novel approach for context-aware resource allocation in D2D-enabled small cell networks. We have formulated the context-aware resource allocation problem as a one-to-one matching game and we have shown that the game exhibits peer effects. To solve the game, we have proposed a distributed social-aware resource allocation algorithm that exploits the physical layer metrics of the wireless network along with the users' social ties from the underlaid social network. Then, we have shown that the proposed algorithm is guaranteed to converge to a two-sided stable matching between the users and the network's resource blocks. Simulation results have shown that the proposed matching-based algorithm yields socially well-connected cluster between D2D links, thus, allowing to offload significantly more traffic than conventional context-unaware approach.
The results provide novel insights into the gains that future wireless networks can achieve from exploiting social context. The results show that with manageable complexity, the proposed context-aware approach can substantially improve the wireless resource utilization by offloading a large amount of traffic from the backhaul-constrained small cell network. This work can be extended to dynamic resource allocation in which the spectrum partitioning may vary, depending on the social tie strength among users.
\vspace{-0cm}
% -------------------------------------------------------------------------
\bibliographystyle{IEEEbib}
\bibliography{references}

\end{document}